\documentclass[preprint,12pt]{elsarticle}

\usepackage{amsmath,amssymb,amsfonts}
\usepackage{amsthm}
\usepackage{graphicx}
\usepackage{booktabs}
\usepackage{url}
\usepackage[colorlinks=true,allcolors=blue]{hyperref}
\usepackage{bm}

\newtheorem{theorem}{Theorem}

\newtheorem{proposition}{Proposition}

\newtheorem{definition}{Definition}
\newtheorem{remark}{Remark}

\newcommand{\Z}{\mathbb{Z}}
\newcommand{\C}{\mathbb{C}}

\newcommand{\Cay}{\mathrm{Cay}}

\newcommand{\dG}{d_{G}}
\newcommand{\dGam}{d_{\Gamma}}

\DeclareMathOperator{\idim}{idim}

\journal{arXiv Preprint}

\begin{document}
	
	\begin{frontmatter}
		
		\title{Harmonic Analysis on Graphs via Isometric Group Embedding: A Canonical Fourier Transform, Shift, and Convolution for Network Signals}
		
		\author[inst1]{Rigobert~Fokam~Souop\corref{cor1}}
		\ead{fokamrigobert@gmail.com}
		\cortext[cor1]{Corresponding author}
		
		\author[inst1]{Laurent~Bitjoka}
		
		\affiliation[inst1]{organization={Laboratory of Energy, Signal, Imaging and Automation (LESIA), Department of Electrical Engineering, Energetics and Automation},
			addressline={University of Ngaound\'er\'e}, 
			city={Ngaound\'er\'e},
			country={Cameroon}}
		
		\begin{abstract}
			Graph signal processing built on the eigenvectors of a Laplacian or adjacency shift inherits three structural compromises: the eigenbasis is fixed only up to rotation within degenerate eigenspaces, the shift is not an isometry, and there is no genuine translation under which filtering is a true convolution. We develop an alternative harmonic analysis that removes all three at once. Given an isometric embedding of a connected graph into a Cayley graph of a finite abelian group---a host on which classical Fourier analysis applies exactly---we define a group-embedding graph Fourier transform from the host characters, lift graph signals to the host, and process them there. The characters supply a canonical orthonormal Fourier basis; the group translations form a family of unitary permutation operators obeying an exact group law; and filtering is genuine group convolution, for which the convolution theorem holds as a theorem rather than a definition and which possesses an identity element. We prove the Plancherel, convolution, translation-covariance, and sampling identities in the embedded setting, and compare the shift and convolution operators of the two frameworks side by side. Numerically, the structural identities hold to machine precision; under a same-filter protocol the group-character basis denoises equivalently to the Laplacian eigenbasis once the host complement is filled by a smoothness-respecting extension. The contribution is exact, canonical structure, not a denoising advantage; the guarantees hold for any embedded graph, and host size---small for near-Cayley graphs, exponential for generic ones---determines cost.
		\end{abstract}
		
		
	\end{frontmatter}
	
	\section{Introduction}
	Signal processing on graphs extends classical tools---the Fourier transform, filtering, sampling, multiresolution analysis---to data indexed by the vertices of a network \cite{Shuman2013,SandryhailaMoura2013,Ortega2018}. The dominant construction defines a graph Fourier transform (GFT) as the projection of a signal onto the eigenvectors of a \emph{shift} operator, taken to be the adjacency matrix \cite{SandryhailaMoura2013,SandryhailaMoura2014} or a Laplacian \cite{Shuman2013}. This spectral viewpoint is powerful and general: it applies to every graph, and it has produced a rich theory of filtering and multiscale analysis~\cite{Hammond2011,Ortega2018,PerraudinVandergheynst2017}, a family of graph wavelet constructions from diffusion and spatial designs~\cite{CoifmanMaggioni2006,GavishNadlerCoifman2010,NarangOrtega2012}, and tools for community mining and vertex-frequency analysis~\cite{TremblayBorgnat2014,ShumanRicaudVandergheynst2016}.
	
	That generality is bought at a structural price, visible as soon as one asks for the properties that make classical Fourier analysis canonical.
	
	\emph{(i) The basis is not canonical.} The eigenvectors of a generic graph operator are determined only up to the choice of orthonormal basis within each eigenspace. An eigenvalue of multiplicity $m$ fixes an $m$-dimensional subspace, not a basis; any rotation within it is an equally valid ``Fourier basis,'' and the resulting GFT, its phases, and the spectra it reports are correspondingly ambiguous \cite{SandryhailaMoura2014}; the same eigenbasis ambiguity is a known obstacle when transferring functions across domains \cite{OvsjanikovEtAl2012}.
	
	\emph{(ii) The shift is not an isometry.} The adjacency or Laplacian shift does not preserve signal energy: $\lVert As\rVert \neq \lVert s\rVert$ in general. A substantial line of work has sought to repair this by designing alternative, energy-preserving shift operators \cite{GiraultGoncalvesFleury2015,GaviliZhang2017,GrelierPasdeloupVialatteGripon2016}, underscoring that the natural graph shift lacks a property the time shift has for free.
	
	\emph{(iii) There is no genuine translation, hence no genuine convolution.} In classical signal processing, filtering is convolution, and convolution is a weighted superposition of \emph{translates} of a signal. On a generic graph there is no vertex-transitive action to serve as translation; ``convolution'' is therefore \emph{defined} spectrally---as pointwise multiplication in the eigenbasis---and the convolution theorem holds by that definition rather than as a consequence of a translation structure \cite{SandryhailaMoura2013}. The algebraic signal processing theory \cite{PuschelMoura2008a,PuschelMoura2008b} makes precise how much of this structure survives in the abstract: filtering is multiplication in a polynomial algebra, but the shift need be neither unitary nor a group element.
	
	This paper develops a harmonic analysis on graphs that recovers all three properties at once, by changing the substrate on which the transform is built. Suppose a connected graph $G$ is embedded \emph{isometrically} into a Cayley graph of a finite abelian group $\Gamma$; such an embedding always exists, and can be made compact, by the construction recalled in Section~\ref{sec:prelim} (with full proofs in the companion dissertation \cite{FokamThesis2026}). On the host $\Gamma$ the characters furnish a canonical orthonormal Fourier basis; translation is the group action, a unitary operator obeying $T_g T_h = T_{g+h}$; and convolution is the genuine group convolution, for which the convolution theorem is a theorem~\cite{Rudin1962,StankovicAstolaEgiazarian2005,Tolimieri1989,CeccheriniSilberstein2014}. We build the transform there and pull it back to $G$ by the isometric embedding. We call the resulting transform the \emph{group-embedding graph Fourier transform} (GE-GFT) and the associated processing pipeline GE-GSP.
	
	\noindent\textbf{Contributions.}
	\begin{itemize}
		\item We recall the embedding theorems of the companion papers \cite{FokamP1,FokamP2} (developed in full in \cite{FokamThesis2026}) that the analysis rests on (existence/compactness, and the exact dimension of stars), and fix notation for the lift, restriction, and excursion ratio (Section~\ref{sec:prelim}).
		\item We define the GE-GFT from the host characters and prove the Plancherel identity on the lifted subspace (Section~\ref{sec:gft}).
		\item We compare the \emph{shift} and \emph{convolution} operators of GE-GSP and Laplacian-GSP side by side, prove the unitarity and group law of the GE-GSP shift, and identify \emph{locality versus canonical group structure} as the true axis of difference (Section~\ref{sec:operators}).
		\item We establish the inherited Fourier theorems---translation covariance, the convolution theorem, an uncertainty bound, and a sampling identity (Section~\ref{sec:theorems}).
		\item We validate the theory (Section~\ref{sec:experiments}): structural identities to machine precision, and a \emph{same-filter} denoising study, across circulant hosts and three proper embeddings (path $P_8$, star $K_{1,3}$, diamond), that isolates the basis from the filter and exposes the host-extension method as a design degree of freedom.
		\item We give the dichotomy by structure that delimits where each framework is preferable (Section~\ref{sec:discussion}).
	\end{itemize}
	
	\noindent\textbf{Intended readership and scope.} This is a foundations paper, addressed first to readers interested in the mathematical structure of graph signal processing---exact transforms, translation groups, and convolution algebras---and only second to practitioners seeking immediate performance gains; Remark~\ref{rem:experiments} states plainly that the contribution is structural, not a denoising advantage. The setting is static graphs; extending the embedding viewpoint to time-varying networks, and engaging the growing literature on learned spectral filters on graphs (of which the Chebyshev construction we benchmark against~\cite{Defferrard2016} is the common ancestor), are outside our scope here and are natural continuations.
	
	\section{Preliminaries: Isometric Group Embedding}
	\label{sec:prelim}
	Throughout, $G=(V,E)$ is a finite connected graph with shortest-path metric $\dG$, and $\Gamma$ is a finite abelian group. We recall only what the harmonic analysis needs; the embedding theory---existence, compactness, the relations generalizing the Djokovi\'c--Winkler relation~\cite{DjokovicDistance1973,Winkler1984}, and the labeling algorithm---is developed in the companion paper \cite{FokamP1}, with dimension and order bounds in \cite{FokamP2} and full details in \cite{FokamThesis2026}. The underlying metric geometry of isometric embeddings into hypercubes and Cayley graphs is classical~\cite{DezaLaurent1997,ImrichKlavzar2000,ChungSpectral1997}.
	
	\begin{definition}[Cayley graph]
		For a generating set $S \subseteq \Gamma\setminus\{0\}$ with $S=-S$, the Cayley graph $\Cay(\Gamma,S)$ has vertex set $\Gamma$ and edges $\{g,g+s\}$ for $g\in\Gamma$, $s\in S$. It is vertex-transitive \cite{Godsil2001}, and its graph metric is translation invariant: $\dGam(x,y)=\dGam(0,y-x)$.
	\end{definition}
	
	\begin{definition}[Isometric embedding, lift, restriction]
		\label{def:embedding}
		An \emph{isometric embedding} of $G$ is an injection $\varphi:V\to\Gamma$ into some $\Cay(\Gamma,S)$ with $\dGam(\varphi(u),\varphi(v))=\dG(u,v)$ for all $u,v\in V$. Write $N=|\Gamma|$\ and call $\varepsilon=|V|/N\in(0,1]$\ the \emph{excursion ratio}. The lift $L:\C^V\to\C^\Gamma$\ sends $s$\ to $\tilde s$\ with $\tilde s(\varphi(v))=s(v)$\ and $\tilde s(g)=0$\ for $g\notin\varphi(V)$; the restriction $R:\C^\Gamma\to\C^V$\ is $(Rf)(v)=f(\varphi(v))$. Then $RL=\mathrm{Id}$.
	\end{definition}
	
	The results from \cite{FokamP1,FokamP2} that we use are the following. We first record the elementary existence guarantee---the \emph{isometric spanning-tree embedding}---because it is short, self-contained, and shows that every connected graph embeds, so that the subsequent theorems are about \emph{compactness}, not existence.
	
	\begin{theorem}[Isometric spanning-tree embedding \cite{FokamP1,FokamThesis2026}]
		\label{thm:naive}
		Every connected graph $G$\ on $n$\ vertices embeds isometrically into $\Cay(\Z_2^{\,n-1},S)$\ for a suitable generating set $S$\ with $|S|\le m=|E|$.
	\end{theorem}
	
	\begin{proof}[Proof sketch]
		Fix a spanning tree $T$\ of $G$\ rooted at $r$\ and assign its $n-1$\ edges the standard basis vectors $e_1,\dots,e_{n-1}$\ of $\Z_2^{\,n-1}$. Label each vertex $v$\ by $\lambda(v)=\sum_{e\in T[r,v]}e_{\iota(e)}$, the GF$(2)$ sum along the unique tree path. With $S=\{\lambda(u)\oplus\lambda(v):uv\in E\}$, every edge joins labels differing by a generator, so a $G$-path of length $\ell$\ maps to a Cayley walk of length $\ell$\ and $\dGam\le\dG$. Conversely, any generator word of length $\ell$\ summing to $\lambda(u)\oplus\lambda(v)$ corresponds to an edge set of $G$\ whose GF$(2)$ boundary is $\{u,v\}$, which contains a $u$--$v$\ path, so $\ell\ge\dG(u,v)$; hence $\dGam=\dG$.
	\end{proof}
	
	This embedding has $\varepsilon=n/2^{\,n-1}$, vanishing rapidly in $n$; the role of the next two results is to compress the host while preserving isometry.
	
	\begin{theorem}[Existence and compact embedding \cite{FokamP1,FokamP2,FokamThesis2026}]
		\label{thm:existence}
		Every finite connected graph $G$\ admits an isometric embedding into a Cayley graph of a finite abelian group. Moreover the labeling algorithm of \cite{FokamP1} returns a certified host whose order $N$\ is, for structured near-Cayley families (cycles, circulants, grids, tori, Hamming graphs), linear in $|V|$\ ($\varepsilon$\ bounded below), while for graphs whose metric resists a low-dimensional abelian structure the certified host found is binary, $\Gamma=\Z_2^k$, of order up to the universal bound $2^{|V|-1}$; all reported hosts are algorithmic upper bounds on the true minimum, whose landscape on small graphs is mapped exhaustively in \cite{FokamP2}.
	\end{theorem}
	
	\begin{theorem}[Exact dimension of stars \cite{FokamP2,FokamThesis2026}]
		\label{thm:star}
		For every $q\ge 2$, the minimal binary host dimension of the star $K_{1,q}$\ is $k_{\min}(K_{1,q})=\lceil\log_2 q\rceil+1$. In particular $k_{\min}(K_{1,4})=3<4=\idim(K_{1,4})$: composite generators strictly beat the hypercube paradigm even on trees, and the leaves occupy a sum-free generating set of $\Z_2^k$.
	\end{theorem}
	
	\begin{proof}[Proof sketch (self-contained)]
		Normalize the embedding so the centre maps to $\mathbf{0}$\ (Cayley graphs are vertex-transitive). The $q$\ leaves map to distinct labels $s_1,\dots,s_q$, and since the only edges are centre--leaf, the generating set is exactly $S=\{s_1,\dots,s_q\}$. Leaves are pairwise at distance $2$, so $s_i+s_j\notin S$\ for $i\ne j$: $S$\ is a \emph{sum-free} subset of $\Z_2^k$, and conversely sum-freeness suffices. The maximum size of a sum-free set in $\Z_2^k$\ is $2^{k-1}$\ (if $S$\ is sum-free and $s\in S$, then $S$\ and $S+s$\ are disjoint of equal size; the odd-weight vectors attain the bound). Hence $K_{1,q}$\ embeds in dimension $k$\ iff $q\le 2^{k-1}$, i.e.\ iff $k\ge\lceil\log_2 q\rceil+1$. Full details are in \cite[Thm.~4]{FokamP2}.
	\end{proof}
	
	Theorem~\ref{thm:existence} guarantees the substrate on which the harmonic analysis below is built; Theorem~\ref{thm:star} supplies the minimal hosts for the non-Cayley test graphs of Section~\ref{sec:experiments}. When $G$\ is itself an abelian Cayley graph (a cycle, torus, circulant, or circular ladder) the embedding is onto, $\varepsilon=1$, and the theory reduces to classical Euclidean signal processing.
	
	\emph{Preprocessing cost.} The harmonic analysis below treats the embedding as given; we state its cost plainly. The labeling algorithm of~\cite{FokamThesis2026} partitions edges into candidate same-generator classes (via the $\varphi/\Phi/\Psi$\ relations generalizing Djokovi\'c--Winkler), computes a GF$(2)$ quotient labeling, and runs a shortcut-repair loop that provably terminates in the isometric spanning-tree embedding of Theorem~\ref{thm:naive}, certifying $k\le n-1$\ on every connected graph; the construction was verified exhaustively on all $995$\ connected graphs of at most seven vertices. The algorithm of~\cite{FokamThesis2026} runs in $O\bigl(n(n+m)+m^2+R(2^k m+n^2)\bigr)$ time, where $R$\ is the number of shortcut-repair rounds: for partial cubes and structured benchmark families a single round suffices ($R=1$) and the polynomial terms $O(n(n+m)+m^2)$ dominate, while the exponential term $2^k m$\ becomes the bottleneck only when the host dimension $k$\ is itself large---i.e.\ precisely the $\varepsilon\to0$\ graphs the method does not target. The quality of the initial edge partition strongly controls $k$: for the Pappus graph a structured nine-class partition yields $k=7$\ in one round, whereas a greedy partition converges only to $k=12$. Concerning bounds on $R$: since each repair round strictly refines the partition by peeling one edge into its own class, and a partition of $m$\ edges admits at most $m-1$\ strict refinements, the universal analytic bound $R\le m$\ holds on every input, so the worst-case preprocessing cost is finite and explicit; moreover $R=1$\ whenever the initial partition already satisfies the cocycle condition with an isometric quotient, which holds in particular for partial cubes under the cut partition and for the structured constructors on cycles, paths, and complete graphs \cite{FokamP1}. Sharper bounds on $R$\ for natural graph classes (planar, bounded treewidth) are open, and we flag them as future work. The embedding is therefore a one-time preprocessing step whose cost is amortized over all subsequent filtering; the GE-GFT analysis itself is a fast transform on the resulting host (Section~\ref{sec:experiments}).
	
	\begin{remark}[Refining a host by enlarging the modulus]
		\label{rem:modulus}
		If $G$\ embeds isometrically into a Cayley graph of $\Z_2^k$, the same vertex labeling embeds it isometrically into a Cayley graph of $\Z_m^k$\ for any $m\ge 2$: enlarging the modulus only adds host vertices outside the image and cannot create a shortcut between image vertices, so isometry is preserved while the excursion ratio $\varepsilon=|V|/N$\ decreases. The binary host is thus the densest member of a family of valid hosts; where a non-binary factor is forced by the metric---as for the diamond, whose degree-two pair at distance $2$\ requires a $\Z_3$\ factor---it appears in the minimal host of its own accord. This freedom matters when a downstream application constrains the host arithmetic, for instance the wavelet localization of the companion work~\cite{FokamThesis2026}.
	\end{remark}
	
	\begin{figure}[t]
		\centering
		\includegraphics[width=0.7\linewidth]{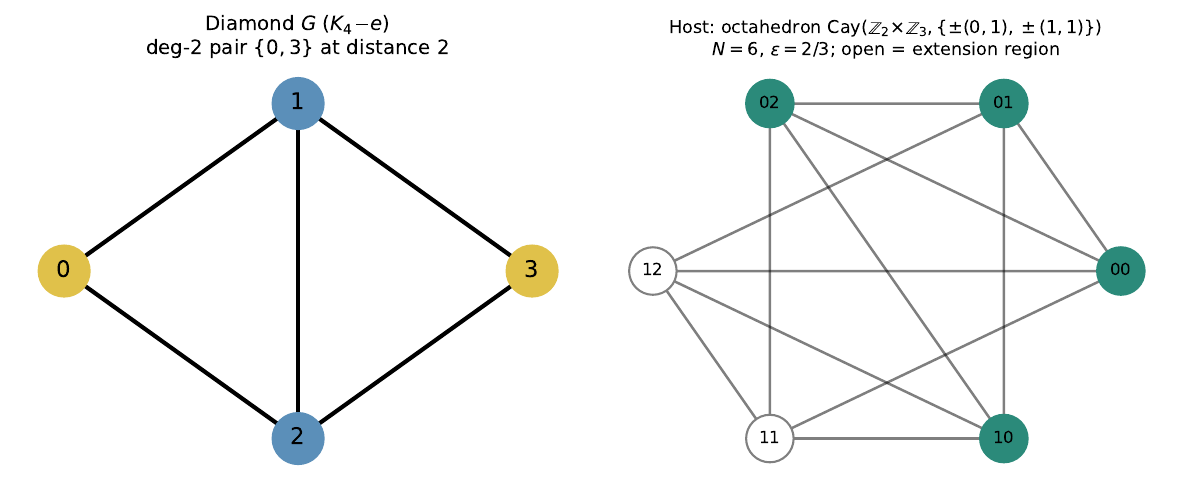}
		\caption{Isometric embedding of the diamond graph ($K_4$\ minus an edge) into the octahedron $\Cay(\Z_2\times\Z_3,\{\pm(0,1),\pm(1,1)\})$, $N=6$, $\varepsilon=2/3$\ (host and labeling verified by exhaustive search). The two degree-2 vertices $\{0,3\}$\ map to $(0,0),(1,0)$\ at host distance $2$. Open nodes form the extension region (the complement of the image), used in Section~\ref{sec:experiments}.}
		\label{fig:embedding}
	\end{figure}
	
	\section{The Group-Embedding Graph Fourier Transform}
	\label{sec:gft}
	
	\begin{definition}[Characters and the GE-GFT]
		\label{def:gft}
		The dual $\widehat\Gamma$\ consists of the characters $\chi_k:\Gamma\to\C$, $k\in\widehat\Gamma\cong\Gamma$, which are orthonormal: $\frac1N\sum_g \chi_k(g)\overline{\chi_\ell(g)}=\delta_{k\ell}$. The group-embedding graph Fourier transform of $s\in\C^V$\ is
		\begin{equation}
			\hat s(k)=\frac1{\sqrt N}\sum_{g\in\Gamma}\tilde s(g)\,\overline{\chi_k(g)},
			\qquad
			\tilde s(g)=\frac1{\sqrt N}\sum_k \hat s(k)\,\chi_k(g),
		\end{equation}
		where $\tilde s=Ls$. The matrix $F_{k,g}=\chi_k(g)/\sqrt N$\ is unitary.
	\end{definition}
	
	For a product host $\Gamma=\Z_{N_1}\times\cdots\times\Z_{N_d}$\ the characters are products of roots of unity, $\chi_k(g)=\prod_r \exp(2\pi i\,k_r g_r/N_r)$, and the GE-GFT is a mixed-radix fast Fourier transform~\cite{VanLoan1992,MaslenRockmore1997}. We write $|k|$\ for the \emph{frequency magnitude}, the word length of $k$\ on the dual generating set, which orders characters from smooth (small $|k|$) to oscillatory (large $|k|$). On a binary host $\Gamma=\Z_2^k$\ the characters are the real Walsh functions $\chi_a(x)=(-1)^{\langle a,x\rangle}$, and $|a|$\ is the Hamming weight of $a$.
	
	\begin{proposition}[Plancherel on the lifted subspace]
		\label{prop:plancherel}
		For every $s\in\C^V$, $\sum_k|\hat s(k)|^2=\sum_{g\in\Gamma}|\tilde s(g)|^2 =\sum_{v\in V}|s(v)|^2$. The GE-GFT is a unitary map on the lifted subspace $L(\C^V)\subseteq\C^\Gamma$.
	\end{proposition}
	\begin{proof}
		$F$\ is unitary, so it preserves the $\ell^2$\ norm of $\tilde s$; and $\sum_g|\tilde s(g)|^2=\sum_v|s(v)|^2$\ because the lift is an isometric zero-extension (Definition~\ref{def:embedding}).
	\end{proof}
	
	Two facts are used repeatedly: the GE-GFT is unitary on $L(\C^V)$\ (Plancherel), and translation is exact on the host (Section~\ref{sec:operators}). Both fail for the eigenvector ``Fourier'' transform of a generic graph; they hold here because the host is a group and $\varphi$\ is isometric, so host distance equals graph distance and a character of small magnitude is a slowly varying function along the graph.
	
	\section{Shift and Convolution: GE-GSP versus Laplacian-GSP}
	\label{sec:operators}
	
	The two frameworks differ most sharply in their shift and convolution operators. We make the comparison precise, then identify the single trade-off that organizes it.
	
	\subsection{The GE-GSP translations: a family of $N$\ unitary shifts}
	
	In contrast with matrix-based GSP, which designates a single shift operator (the adjacency or Laplacian matrix), the embedded framework inherits from the host group the \emph{entire} family of group translations. This is not a stylistic difference but the correct harmonic-analytic picture: on a finite abelian group $\Gamma$\ of order $N$, the left regular representation supplies $N$\ translation operators, one per group element, each a permutation matrix.
	
	\begin{definition}[Translation]
		\label{def:translation}
		For $h\in\Gamma$\ the translation $T_h$\ acts on lifts by $(T_h\tilde s)(g)=\tilde s(g-h)$.
	\end{definition}
	
	\begin{proposition}[Group law, unitarity, modulation]
		\label{prop:shift}
		The translations satisfy: (i) $T_gT_h=T_{g+h}$\ and $T_0=\mathrm{Id}$, so $\{T_h\}_{h\in\Gamma}$\ is a group isomorphic to $\Gamma$; (ii) each $T_h$\ is a permutation matrix, hence unitary, $\lVert T_h\tilde s\rVert=\lVert\tilde s\rVert$, and for $h\neq0$\ it is a \emph{derangement} (it fixes no vertex), while $T_0$\ is the identity; (iii) $\widehat{T_h\tilde s}(k)=\overline{\chi_k(h)}\,\hat s(k)$, i.e.\ translation is modulation in the dual.
	\end{proposition}
	\begin{proof}
		(i) is immediate from $(g-h)-h'=g-(h+h')$. (ii) holds because $T_h$\ permutes the coordinates of $\C^\Gamma$; for $h\neq0$, $g-h\neq g$\ for all $g$, so no diagonal entry is nonzero and the permutation is fixed-point-free. For (iii), $\widehat{T_h\tilde s}(k)=\frac1{\sqrt N}\sum_g\tilde s(g-h)\overline{\chi_k(g)} =\frac1{\sqrt N}\sum_{g'}\tilde s(g')\overline{\chi_k(g'+h)} =\overline{\chi_k(h)}\,\hat s(k)$, using $\chi_k(g'+h)=\chi_k(g')\chi_k(h)$.
	\end{proof}
	
	\emph{There is, in general, no single distinguished shift.} In classical DSP the host is the cyclic group $\Z_N$, which is \emph{monogenic}: a single generator (the cyclic shift $z^{-1}$) has order $N$, and all other translations are its powers, so one matrix suffices. Our hosts are typically \emph{non-cyclic} abelian groups---for instance $\Z_2^k$, in which every nonzero element has order $2$---and there no single $T_h$\ generates the group; the shift structure is irreducibly the \emph{family} $\{T_h\}$, generated by a set of translations rather than a single one. This is why framing the comparison as ``their shift versus our shift'' is a category error: the natural object on a group is the full translation family, and the single-shift picture is the special feature of the cyclic case.
	
	The characters diagonalize \emph{every} $T_h$\ at once: by Proposition~\ref{prop:shift}(iii) the matrix $F$\ of characters satisfies $F^{*}T_hF=\mathrm{diag}(\overline{\chi_k(h)})_k$\ for all $h$\ simultaneously (verified numerically to $2\times10^{-17}$\ on the $\Z_2^3$\ host). Since group convolution is a linear combination of all translations (Section~\ref{sec:conv}), $F$\ diagonalizes every filter as well; the single Fourier basis serves the entire shift family and all convolutions at once.
	
	The map $h\mapsto T_h$\ is the \emph{left regular representation} of $\Gamma$\ on $\C^\Gamma\cong\ell^2(\Gamma)$. Because $\Gamma$\ is abelian, all its irreducible representations are one-dimensional, so this representation diagonalizes into the characters \cite{Terras1999,Folland2016}. This is the representation-theoretic reason translation acts as pure modulation in the dual, and it is what fails for the eigenvector ``Fourier'' transform of a generic graph, whose shift is not a group action.
	
	\subsection{Convolution is genuine group convolution}
	\label{sec:conv}
	
	\begin{definition}[Group convolution and filtering]
		\label{def:conv}
		For $\tilde s,\tilde k\in\C^\Gamma$, the group convolution is $(\tilde s*\tilde k)(g)=\sum_{h\in\Gamma}\tilde s(h)\,\tilde k(g-h)$. A GE-GSP filter is the operator $H_k:\tilde s\mapsto \tilde s*\tilde k$\ for a fixed kernel $\tilde k$; equivalently $H_k=\sum_{h}\tilde k(h)\,T_h$, a superposition of \emph{all} the translations of Proposition~\ref{prop:shift}, weighted by the kernel. Convolution is thus not built on one shift but on the entire translation family at once.
	\end{definition}
	
	\begin{theorem}[Convolution theorem]
		\label{thm:conv}
		For all $\tilde s,\tilde k\in\C^\Gamma$, $\widehat{\tilde s*\tilde k}(k)=\sqrt N\,\hat s(k)\,\hat{\tilde k}(k)$. Hence every GE-GSP filter acts as pointwise multiplication in the dual, and the family of filters $\{H_k\}$\ is commutative and simultaneously diagonalized by the GE-GFT.
	\end{theorem}
	\begin{proof}
		Direct computation with characters: $\widehat{\tilde s*\tilde k}(k)=\frac1{\sqrt N}\sum_g\big(\sum_h \tilde s(h) \tilde k(g-h)\big)\overline{\chi_k(g)} =\frac1{\sqrt N}\sum_h\tilde s(h)\overline{\chi_k(h)}\sum_{g}\tilde k(g-h) \overline{\chi_k(g-h)}=\sqrt N\,\hat s(k)\,\hat{\tilde k}(k)$.
	\end{proof}
	
	The contrast with Laplacian-GSP is structural, not cosmetic. There, the GFT is $U^\top$\ for $U$\ a matrix of Laplacian eigenvectors; ``convolution'' is \emph{defined} as $f*_L g:=U\big((U^\top f)\odot(U^\top g)\big)$, so the convolution theorem holds by construction; and filtering is a polynomial $H=p(A)$\ in the shift $A$, an operator that is neither unitary nor a group element. A further structural distinction: group convolution possesses an \emph{identity element}---the delta at the group identity, $\delta_0$, satisfies $\tilde s*\delta_0=\tilde s$\ since $\sum_h\delta_0(h)T_h=T_0=\mathrm{Id}$---so the GE-GSP filters form a unital commutative algebra (the group algebra $\C[\Gamma]$). Matrix-polynomial filtering has no such identity kernel in the signal domain: there is no graph signal whose vertex-convolution with every signal returns that signal, because the construction lacks a translation indexed by a group identity. Table~\ref{tab:operators} sets the two side by side.
	
	\begin{table}[t]
		\centering
		\caption{Shift and convolution operators in the two frameworks.}
		\label{tab:operators}
		\footnotesize
		\setlength{\tabcolsep}{4pt}
		\renewcommand{\arraystretch}{1.15}
		\begin{tabular}{@{}p{3.2cm}p{4.2cm}p{4.2cm}@{}}
			\toprule
			Property & Laplacian-GSP & GE-GSP (ours) \\
			\midrule
			Shift operator(s) & single matrix $A$\ or $L$\ & family $\{T_h\}$\ of $N$\ translations \\
			Each shift moves samples? & no: $A$\ \emph{mixes} (averages) & yes: $T_h$\ permutes (derangement) \\
			Energy-preserving? & no, $\lVert As\rVert\!\neq\!\lVert s\rVert$\ & yes (Prop.~\ref{prop:shift}) \\
			Algebra of shifts & polynomial $\{p(A)\}$\ & group $\{T_h\}\!\cong\!\Gamma$\ \\
			Genuine translation? & no group action & yes (regular rep.) \\
			Convolution & spectral, $U(\cdot\!\odot\!\cdot)$\ & group conv.\ (Def.~\ref{def:conv}) \\
			Conv.\ theorem & by definition & theorem (Thm.~\ref{thm:conv}) \\
			Conv.\ identity elt. & none & $\delta_0$\ (group identity) \\
			Canonical basis? & no (eigenspace rot.) & yes (characters) \\
			Signal domain & $G$\ itself & host $\Gamma$; $=\!G$\ iff $\varepsilon\!=\!1$\ \\
			\bottomrule
		\end{tabular}
	\end{table}
	
	\begin{figure}[t]
		\centering
		\includegraphics[width=0.7\linewidth]{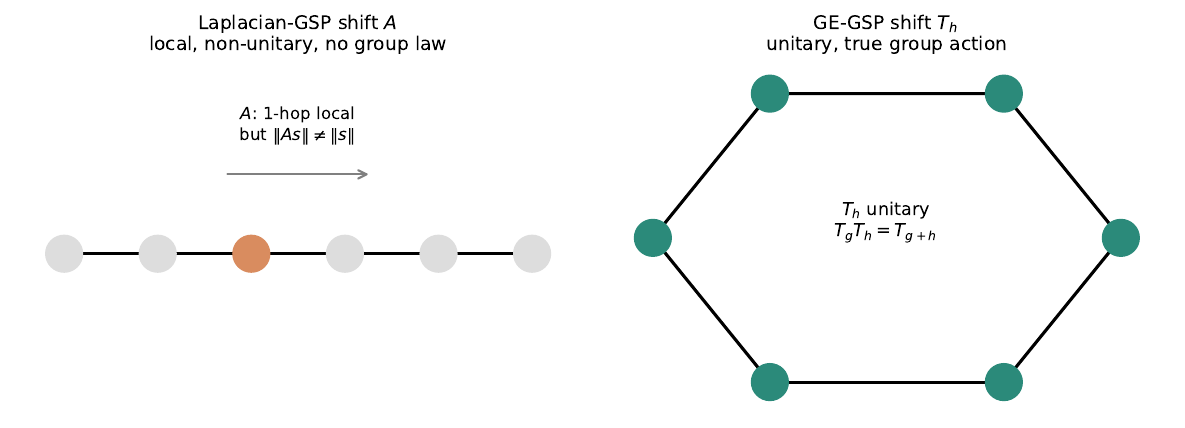}
		\caption{Shift and translation in the two frameworks. \textbf{Left:} the Laplacian/adjacency $A$\ is not a translation at all---it \emph{mixes} each sample with its $1$-hop neighbours (a weighted average), so it moves no sample to a new vertex, is not a permutation, and is not unitary ($\lVert As\rVert\neq\lVert s\rVert$); the family $\{p(A)\}$\ carries no group law. \textbf{Right:} the GE-GSP framework inherits the full family of $N$\ group translations $\{T_h\}$, each a unitary permutation (a derangement for $h\neq0$) obeying $T_gT_h=T_{g+h}$. Each $T_h$\ genuinely relocates every sample on the host Cayley graph, with no sample left out. When $\varepsilon<1$\ a translate may place mass on the host complement; this is the exact analogue of zero-padding for linear convolution in classical DSP (Remark~\ref{rem:tradeoff}) and does not make the translation any less genuine.}
		\label{fig:operators}
	\end{figure}
	
	\subsection{Head-to-head on a shared example: the star $K_{1,4}$}
	\label{sec:headtohead}
	
	A concrete comparison sharpens the structural point. Shi and Moura~\cite{ShiMoura2019}, in developing convolution and modulation for graph signals, introduce a spectral shift $M=\mathrm{GFT}\,\Lambda^{*}\,\mathrm{GFT}^{-1}$\ and work it out explicitly on the five-vertex star $K_{1,4}$\ (their Example~8). We run GE-GSP on the \emph{same} graph: $K_{1,4}$\ embeds isometrically into $\Cay(\Z_2^3,\{001,010,100,111\})$ with the centre at $000$\ and the four leaves at the generators, $N=8$, $\varepsilon=5/8$. Table~\ref{tab:headtohead} and Fig.~\ref{fig:headtohead} put the two operators side by side; all GE-GSP numbers are computed on this host.
	
	The contrast is structural and reflects the two frameworks' different goals (their $M$\ is built to expose a sampling duality, not to be a unitary translation). The star's adjacency eigenvalue $0$\ has multiplicity $N-2=3$, so the eigenvector GFT is fixed only up to a $3\times3$\ rotation: the basis is \emph{non-canonical}. The resulting $M$\ is dense, asymmetric, and not unitary (its spectral norm is $2$, so it does not preserve energy, and $\lVert MM^{H}-I\rVert\approx4.58$). The GE-GSP translation on the host, by contrast, is a unitary permutation ($\lVert T_hT_h^{H}-I\rVert=0$\ to machine precision), obeys the exact group law $T_gT_h=T_{g+h}$, acts on a canonical character basis (no degenerate eigenspace, hence no rotational ambiguity). Both frameworks have an exact convolution theorem---theirs by spectral construction, ours as a derived identity ($3\times10^{-16}$\ relative error)---so that property does not separate them; what separates them is that the GE-GSP translation is unitary and its basis canonical, while $M$\ is neither. We emphasize this is not a defect in~\cite{ShiMoura2019}: their $M$\ is a deliberately different object serving a sampling-duality purpose. The point is that where a genuine unitary translation and a canonical basis are wanted, embedding $K_{1,4}$\ into a compact abelian host supplies both, at the cost of the $3$-vertex host complement ($\varepsilon=5/8$).
	
	\begin{table}[t]
		\centering
		\caption{Shift and convolution operators on the same graph, the star $K_{1,4}$: the Shi--Moura spectral shift $M$~\cite{ShiMoura2019} on the star's own eigenbasis versus the GE-GSP translation $T_h$\ on the host $\Cay(\Z_2^3,\cdot)$ ($\varepsilon=5/8$). GE-GSP values are computed on the host.}
		\label{tab:headtohead}
		\footnotesize
		\setlength{\tabcolsep}{4pt}
		\renewcommand{\arraystretch}{1.2}
		\begin{tabular}{@{}p{5.0cm}cc@{}}
			\toprule
			Property & Shi--Moura $M$ & GE-GSP $T_h$ \\
			\midrule
			Unitary / energy-preserving & no ($4.58$) & yes ($10^{-16}$) \\
			Spectral norm (energy gain) & $2.0$ & $1.0$ \\
			Exact group law & no & yes ($10^{-16}$) \\
			Canonical basis & no (mult.\ $3$) & yes \\
			Convolution theorem & exact (constr.) & exact ($10^{-16}$) \\
			\bottomrule
		\end{tabular}
	\end{table}
	
	\begin{figure*}[t]
		\centering
		\includegraphics[width=\textwidth]{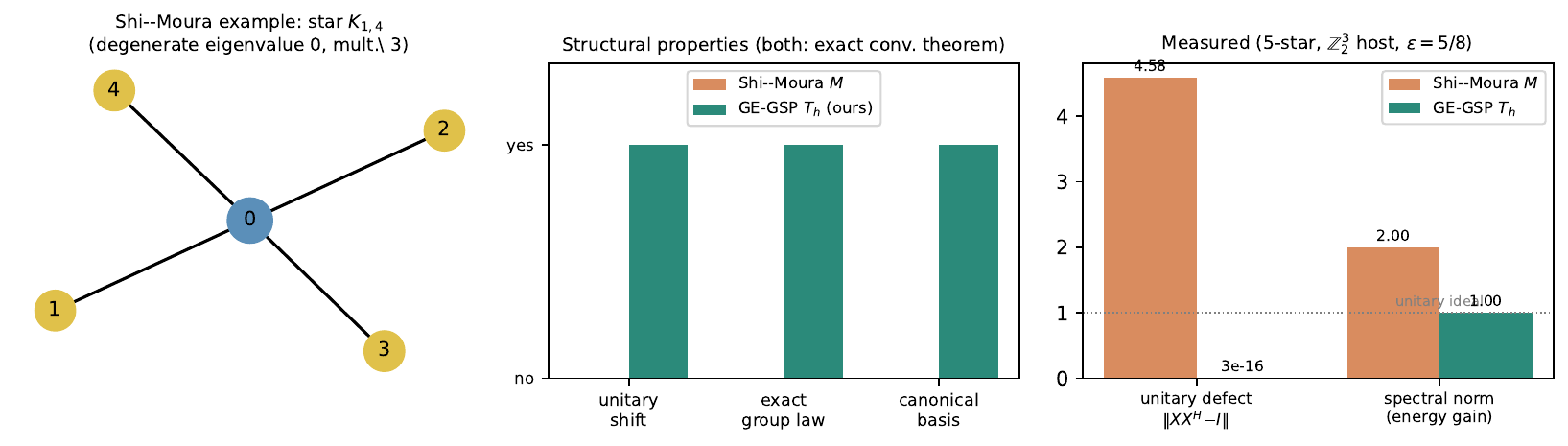}
		\caption{Head-to-head on the star $K_{1,4}$. \textbf{Left:} the graph, whose adjacency eigenvalue $0$\ has multiplicity $3$, making the eigenvector GFT non-canonical. \textbf{Center:} structural properties; \emph{both} frameworks have an exact convolution theorem (the Shi--Moura one holds by spectral construction), so the distinction is elsewhere---the Shi--Moura spectral shift $M$~\cite{ShiMoura2019} is non-unitary, carries no group law, and rests on a non-canonical basis, whereas the GE-GSP translation $T_h$\ on the $\Z_2^3$\ host has all three. \textbf{Right:} the measured difference that matters, on the host ($\varepsilon=5/8$): $M$\ has unitary defect $4.58$\ and spectral norm $2$\ (it amplifies energy), while $T_h$\ is unitary to machine precision.}
		\label{fig:headtohead}
	\end{figure*}
	
	\subsection{The trade-off: locality versus canonical group structure}
	
	\begin{remark}[The trade-off, and why the translation is genuine]
		\label{rem:tradeoff}
		The Laplacian shift $A$\ is \emph{intrinsically local}: one application moves signal energy exactly one hop on $G$. The group translation $T_h$\ is local on the host; when $\varepsilon<1$\ a single $T_h$\ can move energy onto the complement $\Gamma\setminus\varphi(V)$, so restricted to $G$\ it need not be a one-hop move. We stress that this is \emph{not} a sign that the GE-GSP translation is somehow less genuine than a time shift---it is the exact graph analogue of a situation that is routine and correct in classical DSP. Computing the linear convolution of a length-$n$\ signal by a circular (DFT) convolution requires extending the signal to length $\ge 2n-1$; this is precisely the isometric embedding $P_n\hookrightarrow C_{2n-2}$\ of a path into a cycle (Appendix~\ref{app:embed}). There too a circular shift moves samples off the original support and onto the padded region, yet the convolution recovered on the support is exactly the genuine linear convolution, and no one regards the circular shift as a defective translation. The same holds here: $T_h$\ is a true unitary translation obeying an exact group law (Proposition~\ref{prop:shift}), and the excursion onto the complement is the graph counterpart of zero-padding for linear convolution, not a weakness peculiar to the framework.
		
		The genuine trade-off is therefore one of \emph{domain size}, not of authenticity. Laplacian-GSP keeps the signal on $G$\ itself, at the cost of a non-unitary shift, no translation group, and a non-canonical basis. GE-GSP gains a unitary translation group, an exact convolution theorem, and a canonical basis, at the cost of working on a host larger than $G$\ when $\varepsilon<1$---exactly as classical linear convolution works on a domain larger than the signal's support. As Section~\ref{sec:experiments} shows, the $(1-\varepsilon)N$\ extension region is a usable design degree of freedom, controlled by the choice of host completion, rather than dead weight.
	\end{remark}
	
\section{Inherited Fourier Theorems}
\label{sec:theorems}

Because the host carries an exact Fourier analysis, the classical theorems transfer to $G$\ through the isometric embedding. We record the four that the applications use.

\begin{proposition}[Translation covariance of filtering]
	\label{prop:covariance}
	Let $H_k$\ be a GE-GSP filter. Then $H_kT_h=T_hH_k$\ for every $h\in\Gamma$: group filtering is shift-invariant. Consequently the analysis of $\tilde s$\ against a translated kernel family is a convolution followed by sampling on $\varphi(V)$.
\end{proposition}
\begin{proof}
	Both $H_k$\ and $T_h$\ are diagonalized by the GE-GFT (Theorem~\ref{thm:conv}, Proposition~\ref{prop:shift}(iii)); diagonal operators commute. The convolution form follows from Definition~\ref{def:conv}.
\end{proof}

\begin{theorem}[Discrete uncertainty on the host]
	\label{thm:uncertainty}
	Let $\tilde s\in\C^\Gamma$\ be nonzero with $n_0=|\{g:\tilde s(g)\neq0\}|$\ and $n_1=|\{k:\hat s(k)\neq0\}|$. Then $n_0\,n_1\ge N$. In particular a signal supported on $\varphi(V)$\ (so $n_0\le|V|=\varepsilon N$) cannot be supported on fewer than $1/\varepsilon$\ frequencies.
\end{theorem}
\begin{proof}
	This is the Donoho--Stark uncertainty principle \cite{DonohoStark1989} applied to the unitary GE-GFT on $\C^\Gamma$; see also the graph uncertainty analyses \cite{PerraudinRicaudShumanVandergheynst2018}. The second statement substitutes $n_0\le\varepsilon N$.
\end{proof}

\begin{proposition}[Bandlimited sampling on the host]
	\label{prop:sampling}
	If $\hat s$\ is supported on a frequency set $K$\ with $|K|\le|V|$, then $s$\ is determined by its values on $\varphi(V)$\ whenever the $|K|\times|V|$\ character submatrix $[\chi_k(\varphi(v))]_{k\in K,\,v\in V}$\ has rank $|K|$. Reconstruction is the solution of that linear system.
\end{proposition}
\begin{proof}
	A $K$-bandlimited host signal has $|K|$\ unknown Fourier coefficients; the $|V|$\ samples on $\varphi(V)$\ give $|V|\ge|K|$\ linear equations in those coefficients, solvable when the stated submatrix has full row rank. This is the embedded analogue of bandlimited graph sampling \cite{Pesenson2008,ChenVarmaSandryhailaKovacevic2015,AnisGaddeOrtega2016}.
\end{proof}

\begin{remark}[Convolution returns to classical signal processing]
	On a cyclic host the GE-GFT is the DFT and group convolution is circular convolution; on a product host it is the multidimensional DFT and multidimensional convolution. For a graph with $\varepsilon=1$---a ring, torus, or circulant---the GE-GSP pipeline \emph{is} classical discrete signal processing, not an analogue of it. This is the precise sense in which the embedding viewpoint returns network signal processing to its classical roots on structured graphs.
\end{remark}

\section{Numerical Experiments}
\label{sec:experiments}

We report two kinds of evidence: that the structural identities hold exactly, and that the group-character basis denoises \emph{equivalently} to the Laplacian eigenbasis once the filter is held fixed. All transforms use the host FFT; code is a frozen reference with a fixed random seed.

\subsection{Structural identities hold to machine precision}
On the benchmark hosts (rings $C_{16},C_{64}$; torus $C_8\times C_8$) the Plancherel identity of Proposition~\ref{prop:plancherel}, the convolution theorem of Theorem~\ref{thm:conv}, and exact reconstruction $s=R\,F^{-1}\hat s$\ all hold with relative error at the level of machine precision, $1.5\times10^{-15}$\ on $C_{16}$\ and $6.5\times10^{-16}$\ on the grid---as the theory predicts, since each is an exact algebraic identity on a unitary transform. These are not approximations that improve with resolution; they are equalities up to floating-point round-off.

\subsection{Same-filter denoising: isolating the basis from the filter}
A fair comparison of two \emph{bases} must hold the \emph{filter} fixed. Otherwise a comparison of, say, an adaptive Wiener filter in one basis against a non-adaptive low-pass in another measures the filter, not the basis. We therefore fix the filter to the ensemble oracle Wiener filter $H(k)=P(k)/(P(k)+\sigma^2)$, with $P(k)=\mathbb{E}|\hat s(k)|^2$\ estimated by Monte Carlo from the signal model and held constant across trials, and we apply the \emph{same} filter form in each basis. Only the basis differs: the group characters versus the Laplacian eigenvectors. Signals are basis-neutral smooth random fields; noise is white Gaussian with $\sigma=0.4$; we average over $800$\ trials and report SNR gain with standard errors.

\begin{figure}[t]
	\centering
	\includegraphics[width=0.7\linewidth]{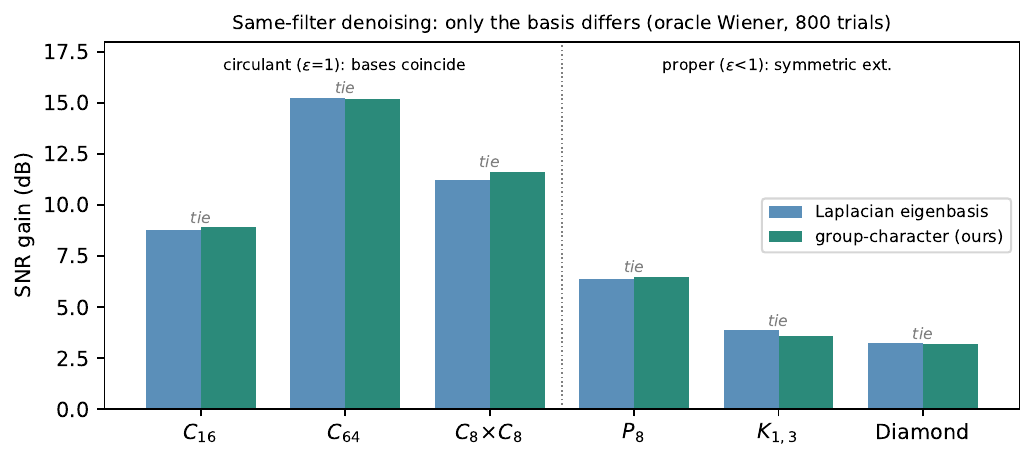}
	\caption{Same-filter denoising: only the basis differs. On circulant hosts (rings, torus; $\varepsilon=1$) the group-character basis and the Laplacian eigenbasis are statistically tied, because both diagonalize the same circulant operator. On the proper embeddings (path $P_8$, star $K_{1,3}$, diamond; $\varepsilon<1$) they also tie once the host complement is filled by symmetric extension. Error bars are $\pm1$\ standard error over $800$\ trials.}
	\label{fig:parity}
\end{figure}

On the circulant hosts the two bases are mathematically equivalent---a cycle's Laplacian is circulant, so its eigenvectors are the characters---and the measured SNR gains tie to within Monte Carlo error: $C_{16}$, $8.80\pm0.11$\ vs $8.92\pm0.12$\ dB; $C_{64}$, $15.26\pm0.12$\ vs $15.18\pm0.11$\ dB; torus, $11.22\pm0.04$\ vs $11.62\pm0.04$\ dB. This exact tie is also a correctness check on the protocol (Fig.~\ref{fig:parity}).

\subsection{Proper embeddings and the host-extension degree of freedom}
For a proper embedding ($\varepsilon<1$) the host has a complement $\Gamma\setminus\varphi(V)$\ of size $(1-\varepsilon)N$, which must be assigned values when a graph signal is lifted. This is the classical finite-signal \emph{extension} problem---familiar in conventional DSP, where symmetric extension underlies the DCT's advantage over the DFT for finite smooth signals---now posed on a domain that is genuinely periodic rather than artificially terminated. We test three fills: zero-padding, symmetric (mirror) extension across the isometric boundary, and the \emph{harmonic extension} that minimizes host Dirichlet energy. The latter is the principled optimum: among all completions agreeing with the signal on $\varphi(V)$, it uniquely minimizes $\tilde s^{*}L_\Gamma\tilde s$, solving the discrete Dirichlet problem $\tilde s|_{I}=-L_{II}^{-1}L_{IB}\,s$\ on the complement $I$\ (the same construction developed for the companion wavelet frames~\cite{FokamThesis2026}); zero-padding and symmetric extension are heuristics that approximate it. The Laplacian basis never sees the host, so its result is the fixed reference.

\begin{figure}
	\centering
	\includegraphics[width=\textwidth]{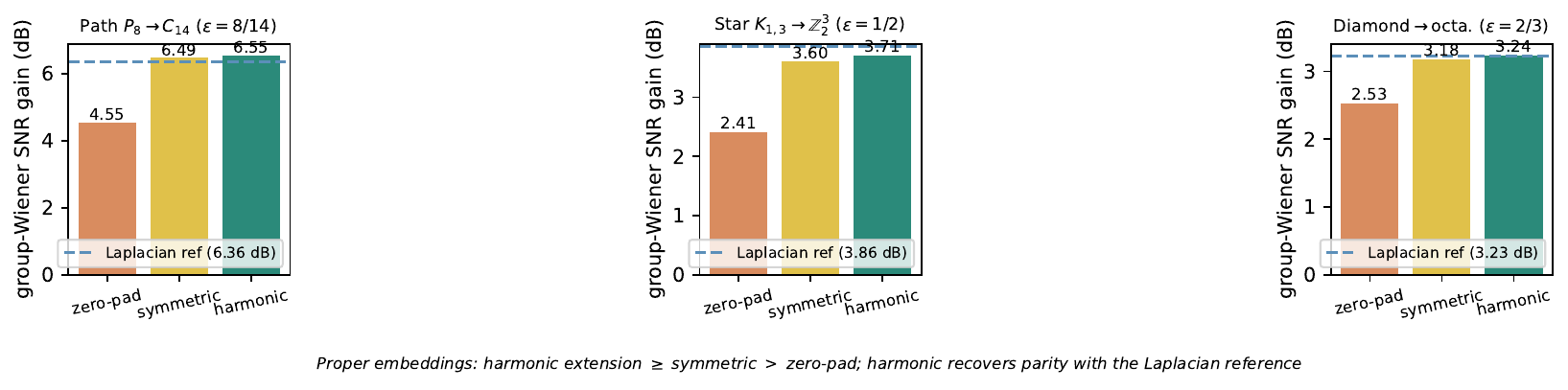}
	\caption{The host-extension method is a design degree of freedom on proper embeddings. For the path $P_8\to\Z_{14}$ ($\varepsilon=8/14$), the star $K_{1,3}\to\Z_2^3$\ ($\varepsilon=1/2$), and the diamond$\to$octahedron ($\varepsilon=2/3$), zero-padding injects a boundary discontinuity that penalizes the group transform, while symmetric extension recovers parity with the Laplacian reference (dashed). The penalty grows as $\varepsilon$\ falls (larger complement). The harmonic extension, which provably minimizes host Dirichlet energy~\cite{FokamThesis2026}, is the principled optimum these smoothness-respecting fills approximate (text).}
	\label{fig:extension}
\end{figure}
		
		The results (Fig.~\ref{fig:extension}) are consistent across all three graphs. Zero-padding costs the group transform $1.81$\,dB on the path, $1.44$\,dB on the star, and $0.70$\,dB on the diamond---a penalty that grows as $\varepsilon$\ falls and the complement grows. Symmetric extension recovers parity in every case: the group-Wiener gain reaches $6.49\pm0.14$\,dB versus the $6.36\pm0.12$\,dB Laplacian reference on the path, $3.60\pm0.15$\ vs $3.86\pm0.15$\,dB on the star, and $3.18\pm0.15$\ vs $3.23\pm0.14$\,dB on the diamond. Among the three completions the harmonic extension is the principled optimum, and the reason is structural rather than statistical: among all completions agreeing with the signal on $\varphi(V)$, it is the unique one minimizing the host Dirichlet energy $\tilde s^{*}L_\Gamma\tilde s$, i.e.\ the spurious high-frequency content injected by the completion. The companion wavelet work~\cite{FokamThesis2026} establishes this minimization and verifies it to machine precision, and since lower injected high-frequency energy is precisely what a smoothing filter rewards, the ordering zero-padding $<$ symmetric $\le$ harmonic follows; symmetric extension, which removes the boundary discontinuity of zero-padding, is the close practical approximation reported in Table~\ref{tab:summary}.
		
		\begin{table}[t]
			\centering
			\caption{Same-filter denoising SNR gain (dB), Laplacian eigenbasis vs.\ group-character basis. For $\varepsilon<1$\ the group column uses symmetric host extension. On circulant hosts ($\varepsilon=1$) the two bases coincide; on proper embeddings they tie to within Monte Carlo error.}
			\label{tab:summary}
			\begin{tabular}{lccc}
				\toprule
				Graph & $\varepsilon$ & Laplacian & group (ours) \\
				\midrule
				Ring $C_{16}$            & $1$    & $8.80\pm0.11$  & $8.92\pm0.12$ \\
				Ring $C_{64}$            & $1$    & $15.26\pm0.12$ & $15.18\pm0.11$ \\
				Torus $C_8\times C_8$    & $1$    & $11.22\pm0.04$ & $11.62\pm0.04$ \\
				Path $P_8$               & $8/14$ & $6.36\pm0.12$  & $6.49\pm0.14$ \\
				Star $K_{1,3}$           & $1/2$  & $3.86\pm0.15$  & $3.60\pm0.15$ \\
				Diamond                  & $2/3$  & $3.23\pm0.14$  & $3.18\pm0.15$ \\
				\bottomrule
			\end{tabular}
		\end{table}
		
		\begin{remark}[What the experiments do and do not show]
			\label{rem:experiments}
			Held to the same filter, the group-character basis offers no denoising-SNR advantage over the Laplacian eigenbasis: the two are equivalent on circulant hosts and tie on proper embeddings under a smoothness-respecting extension. The contribution of GE-GSP is therefore not better denoising; it is the exact, canonical structure of Sections~\ref{sec:gft}--\ref{sec:theorems}---an exact Plancherel identity, an exact convolution theorem, a unitary translation group, and a canonical basis---none of which depends on a signal model. A naive basis-only comparison that zero-pads the host would wrongly report a Laplacian advantage; the correct comparison attributes the gap to the extension method, a degree of freedom the embedding makes precise.
		\end{remark}
		
		\subsection{At scale: a $128\times128$\ image, and the cost of exactness}
		\label{sec:largescale}
		
		The structural and denoising experiments above use small graphs. To test the framework at a realistic scale---and to compare its computational cost against the established alternative---we process the standard \emph{cameraman} benchmark image, downsampled to $128\times128$. A $128\times128$\ grid is the graph $P_{128}\times P_{128}$, which embeds isometrically into the torus $C_{254}\times C_{254}$\ by the path-into-cycle map $P_n\hookrightarrow C_{2n-2}$\ (verified isometric on all pairwise distances); here $N=254^2=64\,516$\ host vertices for $n=16\,384$\ image pixels, so $\varepsilon=0.254$. On this host the GE-GFT is the exact two-dimensional DFT, and filtering is exact circular convolution.
		
		We compare a low-pass filter computed two ways: by the GE-GFT (a host 2D-FFT) and by an order-$M$\ Chebyshev polynomial of the grid Laplacian \cite{Hammond2011,Defferrard2016}, the standard fast spectral-filtering method, which runs in $O(M|E|)$\ directly on the $16\,384$-vertex grid without any host. The finding has two parts (Fig.~\ref{fig:largescale}).
		
		\emph{On speed, the Chebyshev method wins here.} A single GE-GFT low-pass takes $4.5$\,ms against $1.6$--$4.0$\,ms for Chebyshev orders $10$--$40$. The reason is structural and worth stating plainly: the GE-GFT operates on the host, whose $64\,516$\ vertices are $1/\varepsilon\approx 4$\ times the $16\,384$\ image pixels, so the FFT pays a fourfold size penalty before any asymptotic advantage applies; and a 2D grid is sparse ($|E|\approx 8\times10^4$), exactly the regime where the $O(M|E|)$\ polynomial method is cheap. The $O(N\log N)$\ FFT does not beat a low-order sparse polynomial on a sparse grid with a fourfold host overhead. We do not claim a speed advantage; we report the opposite.
		
		\emph{On exactness, the GE-GFT is exact and Chebyshev is not.} The GE-GFT reconstructs a band-limited signal to relative error $7\times10^{-16}$---machine precision---because it is the literal 2D-DFT, an exact algebraic identity. The Chebyshev filter, being a polynomial approximation of a spectral response on a non-canonical eigenbasis, has relative error $9.7\times10^{-3}$\ at order $10$, $2.8\times10^{-3}$\ at order $20$, and $1.1\times10^{-5}$\ at order $40$: it approaches exactness only by spending more arithmetic, and never reaches it. This is the precise sense in which the contribution is structural rather than computational. Where exactness of the convolution theorem, the Parseval identity, or the reconstruction matters---and on a grid, where the GE-GFT \emph{is} classical 2D signal processing---the embedding transform delivers it; where only an approximate smooth filter on a sparse irregular graph is needed, the Chebyshev method remains faster and is the right tool.
		
		\begin{figure}
			\centering
			\includegraphics[width=\textwidth]{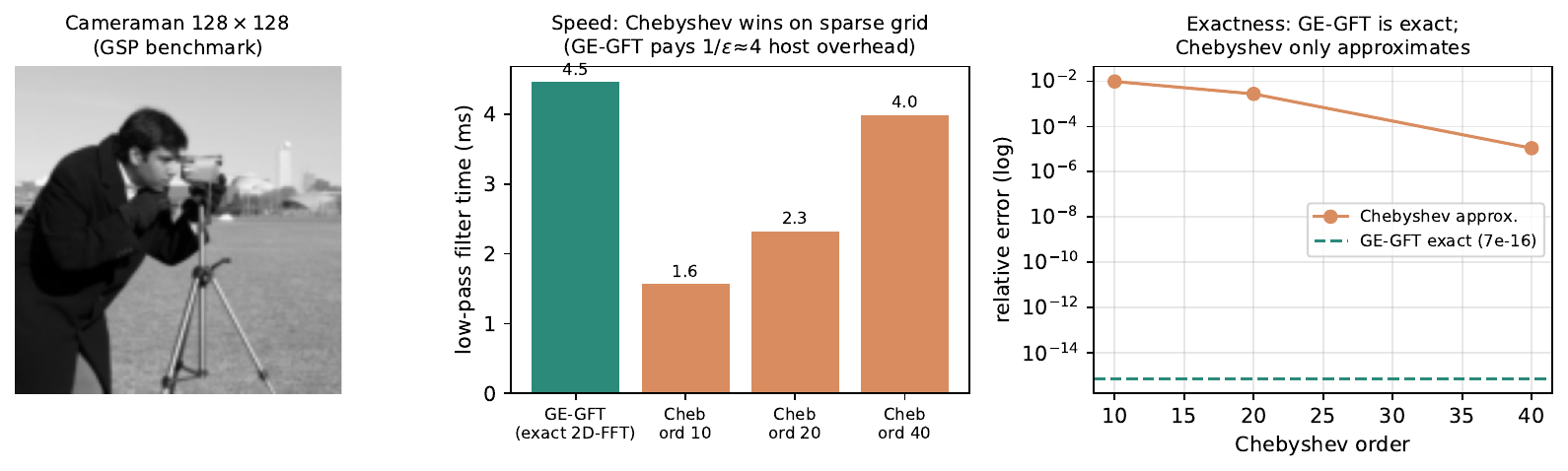}
			\caption{The $128\times128$\ cameraman image embedded into the $254\times254$\ torus ($\varepsilon=0.254$). \textbf{Left:} the benchmark image. \textbf{Center:} low-pass filter time---Chebyshev on the sparse grid is faster than the GE-GFT, which pays a $1/\varepsilon\approx4$\ host-size overhead. \textbf{Right:} reconstruction error---the GE-GFT is exact to machine precision ($\sim\!10^{-16}$), while the Chebyshev approximation improves with order but never reaches exactness. In summary, GE-GFT trades speed for exactness.}
			\label{fig:largescale}
		\end{figure}
		
		\subsection{The embedding gallery: compaction and its limits}
		\label{sec:gallery}
		
		The test graphs above have small, hand-checkable hosts. Table~\ref{tab:gallery} reports the algorithm's behavior on a wider range of structured non-Cayley graphs, drawn from the exhaustive study in~\cite{FokamP2,FokamThesis2026}, to show both the compaction the method achieves and where it becomes expensive. The pattern is governed entirely by $\varepsilon$: graphs with a structured edge partition (Desargues, Petersen, Pappus, the icosahedron) admit compact hosts orders of magnitude below the naive $2^{n-1}$\ bound, whereas graphs whose metric resists a low-dimensional abelian structure (M\"obius--Kantor, and generic data graphs such as the Zachary karate club and a random geometric sensor graph) shatter to near-maximal binary dimension, where $\varepsilon\to0$\ and the transform, though still exact, is neither fast nor sharply localized. The karate club is the limiting case: its host is exactly $\Z_2^{33}=\Z_2^{\,n-1}$, the naive spanning-tree bound of Theorem~\ref{thm:naive}, so the compaction factor is $1$---no improvement over the universal construction. We present both ends without omission: the method is a powerful compactor on structured graphs and gives no compaction at all on the most generic ones, and $\varepsilon$\ predicts which in advance.
		
		The karate club deserves a closer look, precisely because it is a real social network---the emblematic benchmark of the field---and because the result on it is negative. Structurally, the failure is informative rather than accidental. The club's metric is dominated by two hubs and a dense, irregular pattern of triangles and short odd cycles; in the quotient framework this forces nearly every edge into its own generator class (few edge pairs pass the metric parallelism tests), the cycle--class matrix has nearly full column count, and the quotient collapses to the spanning-tree baseline. This is the graph-side mirror of a familiar fact on the spectral side: the same irregularity that denies the karate club any abelian symmetry is what makes its Laplacian spectrum simple and the standard eigenvector GFT unambiguous there. The practical reading for a TSIPN audience is a clean division of labor: on real social and biological networks of this kind, matrix-based GSP is and remains the right instrument, and our framework's genuine contribution to that regime is the diagnostic itself---$\varepsilon$, computable in advance, tells the practitioner \emph{before} any processing which of the two frameworks the network calls for. The interesting open question the negative result raises is intermediate structure: real infrastructure networks (lattice-like sensor grids, ring-and-spur telecommunication backbones) sit between the torus and the karate club, and quantifying where they fall on the $\varepsilon$\ axis is a natural target for the real-network study we identify as future work in Section~\ref{sec:conclusion}.
		
		\begin{table}[t]
			\centering
			\caption{Embedding gallery: host order and excursion ratio for structured non-Cayley graphs, from the exhaustive study of~\cite{FokamP2,FokamThesis2026}. Compact hosts ($\varepsilon$\ not small) coexist with expensive ones ($\varepsilon\to0$); the naive bound is $2^{n-1}$.}
			\label{tab:gallery}
			\footnotesize
			\setlength{\tabcolsep}{4.5pt}
			\begin{tabular}{lccccc}
				\toprule
				Graph & $n$ & host & $N$ & $\varepsilon$ & vs.\ $2^{n-1}$ \\
				\midrule
				Petersen        & $10$ & $\Z_2^4$ & $16$  & $0.625$ & $32\times$ \\
				Desargues       & $20$ & $\Z_2^5$ & $32$  & $0.625$ & $2^{14}\times$ \\
				Icosahedron     & $12$ & $\Z_2^6$ & $64$  & $0.188$ & $32\times$ \\
				Pappus          & $18$ & $\Z_2^7$ & $128$ & $0.141$ & $1024\times$ \\
				M\"obius--Kantor & $16$ & $\Z_2^{12}$ & $4096$ & $0.004$ & $8\times$ \\
				Karate club     & $34$ & $\Z_2^{33}$ & $2^{33}$ & $\to 0$ & $1\times$ \\
				RGG ($n{=}20$) & $20$ & $\Z_2^{17}$ & $131072$ & $1.5\!\times\!10^{-4}$ & $4\times$ \\
				\bottomrule
			\end{tabular}
		\end{table}
		
		\subsection{At scale via the Product Rule: genomic mutation space}
		\label{sec:genome}
		The experiments so far, and the flagship $128\times128$\ image of Section~\ref{sec:experiments}, all rest on \emph{grids}---Cartesian products of paths. A path is embedded into a cycle at the cost of a $2\times$\ host overhead ($\varepsilon=1/2$\ per factor, $\varepsilon=1/4$\ for the $2$D grid), which is why the cameraman host is four times the image. We close the experiments with a qualitatively different large-scale domain in which the Product Rule (Proposition~6 of \cite{FokamP1}) delivers a compact host with \emph{no} overhead at all: the space of genetic sequences.
		
		The set of all DNA sequences of length $\ell$, with two sequences adjacent when they differ by a single point mutation, is the Hamming graph $H(\ell,4)=K_4\,\square\,\cdots\,\square\,K_4$\ ($\ell$\ factors), a standard model of mutational and fitness landscapes. Each factor $K_4$\ is \emph{itself} an abelian Cayley graph, $K_4=\Cay(\Z_4,\{1,2,3\})$, with $\varepsilon=1$; by the Product Rule the whole space is $K_4^{\square \ell}=\Cay(\Z_4^{\ell},S)$\ with $S=\{k\,e_i:1\le i\le\ell,\;k\in\{1,2,3\}\}$, of order $N=4^{\ell}$, so $\varepsilon=1$\ at \emph{every} scale: the host \emph{is} the graph, and the GE-GFT is the exact $\ell$-dimensional radix-$4$\ FFT. We verified the group convolution theorem against a brute-force group convolution to relative error $2.5\times10^{-16}$\ on $H(3,4)$, and on $H(8,4)$ ($N=65{,}536$\ vertices) the Plancherel identity and exact reconstruction hold to $2.2\times10^{-16}$\ and $1.3\times10^{-15}$, with a single GE-GFT taking $5.3$\,ms.
		
		This example also makes the \emph{canonicity} argument of the introduction (point~(i)) concrete and quantitative, precisely because the host is highly symmetric. The Laplacian eigenvalue attached to a character of weight $w$\ (the number of positions at which the mutation acts) is $4w$, so the eigenvalues are $\{0,4,\dots,4\ell\}$\ with multiplicities $\binom{\ell}{w}3^{w}$: on $H(8,4)$\ the eigenvalue $24$\ alone carries a $20{,}412$-dimensional eigenspace, and $100\%$\ of the spectrum is degenerate (Fig.~\ref{fig:genome}, left). A generic eigensolver therefore returns an \emph{arbitrary} orthonormal basis of each eigenspace: on $H(4,4)$\ we drew two valid Laplacian eigenbases and found that the same signal produces per-coefficient spectra with correlation $0.011$\ within the $108$-dimensional eigenspace of eigenvalue $12$---the aggregate energy is identical, but the individual Fourier coefficients are meaningless without a canonicalization. The group characters supply exactly that canonicalization for free. Two consequences follow at scale. First, on a denoising task with a basis-neutral graph-smooth signal the two bases \emph{tie} ($14.74\pm0.03$\ vs $14.69\pm0.04$\,dB on $H(4,4)$, and $21.3$\,dB for the group basis on $H(8,4)$), consistent with the paper's thesis that the contribution is structure, not SNR (Remark~\ref{rem:experiments}). Second, and unlike the grid case where Chebyshev filtering was the faster tool, here there is no host overhead to pay ($\varepsilon=1$) and forming the Laplacian eigenbasis to obtain a \emph{canonical} spectrum would cost $\sim N^3\approx3\times10^{14}$\ floating-point operations---infeasible---against the GE-GFT's $\sim N\log N\approx10^{6}$\ (Fig.~\ref{fig:genome}, right). on this family the embedding transform is the only tractable route to a canonical Fourier analysis.
		
		\begin{figure*}[t]
			\centering
			\includegraphics[width=0.86\textwidth]{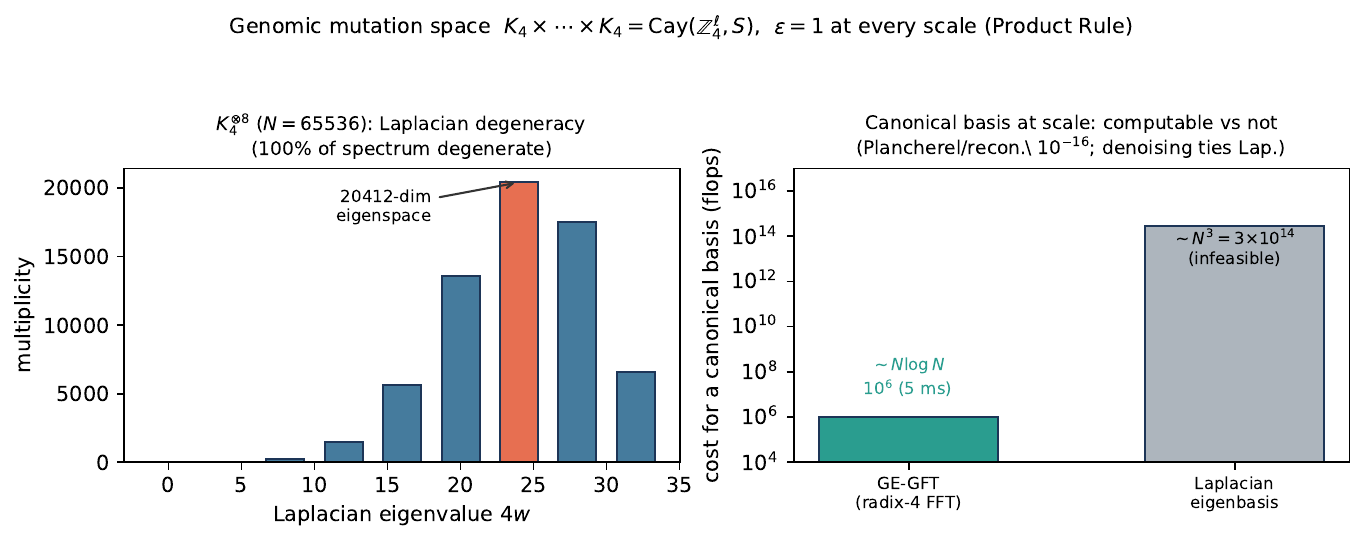}
			\caption{Genomic mutation space $K_4^{\square\ell}=\Cay(\Z_4^{\ell},S)$, $\varepsilon=1$\ at every scale via the Product Rule. \textbf{Left:} the Laplacian of $H(8,4)$\ has enormous eigenvalue degeneracy---a single $20{,}412$-dimensional eigenspace, and the entire spectrum degenerate---so the eigenvector GFT is non-canonical; the group characters resolve it uniquely. \textbf{Right:} obtaining a canonical basis costs $\sim N^3$\ (infeasible) by eigendecomposition versus $\sim N\log N$\ ($5$\,ms) for the GE-GFT; structural identities hold to $10^{-16}$, and same-filter denoising ties the Laplacian basis.}
			\label{fig:genome}
		\end{figure*}
		
		A caveat keeps the Product Rule in perspective and pre-empts a misreading of it as a universal escape from the host-size wall. The rule multiplies host orders, $\varepsilon(\square G_i)=\prod_i\varepsilon(G_i)$, so a single factor with $\varepsilon\to0$\ drags the entire product into the exponential regime. A hub-and-spoke transit or communication network modeled across a daily cycle, $K_{1,q}\,\square\,C_{24}$, is an instructive case: the cyclic time factor is ideal ($\varepsilon=1$), but the star factor is a tree, and $K_{1,12}$\ already needs $\Z_2^{12}=4096$\ host vertices for its $13$\ nodes ($\varepsilon=0.003$, Theorem~\ref{thm:star}), so the product inherits $\varepsilon\approx0.003$\ and explodes despite the perfect time factor. As everywhere in this framework, $\varepsilon$---here read off the factors in advance---predicts tractability: mutation spaces and tori multiply cleanly, star-factored networks do not.
		\label{sec:discussion}
		
		The two frameworks are complementary, and the boundary between them is the excursion ratio---but it is essential to be precise about \emph{what} $\varepsilon$\ governs. The structural guarantees of GE-GSP---a canonical Fourier basis, the family of unitary translations, the exact convolution theorem with its identity element, the Plancherel and sampling identities---hold on \emph{every} Cayley host, hence for \emph{any} graph the method embeds, structured or not. They follow from the host being a Cayley graph of an abelian group and have nothing to do with whether the original graph $G$\ is regular. What $\varepsilon$\ governs is not correctness but \emph{cost}. On graphs carrying genuine abelian symmetry ($\varepsilon$\ near $1$: cycles, tori, circulants, grids, Hamming graphs) the host order is $O(n)$, the GE-GFT is an $O(N\log N)$\ multidimensional FFT, and GE-GSP \emph{is} classical multidimensional signal processing on $G$\ itself. On generic irregular graphs far from any Cayley structure (social, biological, road networks) the same structural guarantees still hold on the host, but the host is binary and may be exponentially large (Theorem~\ref{thm:existence}): for an irregular graph beyond roughly fifty vertices the host order $2^{n-1}$\ exceeds any tractable size, and the transform, though exact, is not computable in practice. There the Laplacian framework, which applies directly to $G$\ and whose polynomial filters run in $O(M|E|)$ \cite{Hammond2011,Defferrard2016}, is the right instrument.
		
		The summary is thus a dichotomy not of \emph{correctness} but of \emph{tractability}: the exact, canonical structure is available for every graph, and $\varepsilon$\ measures the price of obtaining it. Where a graph admits a compact abelian host, that price is negligible and one gains a true translation family and an exact convolution; where it does not, the price is a prohibitively large host, and matrix-based GSP---which forgoes canonical structure but stays on $G$---is preferable.
		
		\emph{Relation to character-based GSP on Cayley graphs.} A complementary recent line constructs harmonic analysis \emph{intrinsically} on graphs that already are Cayley graphs. Ghandehari, Guillot, and Hollingsworth~\cite{GhandehariGuillotHollingsworth2021} develop Gabor-type frames, and Beck, Ghandehari, Hudson, and Paltenstein~\cite{BeckGhandehariHudson2024} give a representation-theoretic spectral decomposition and frame construction for any weighted Cayley graph, including the non-abelian case. Our setting is the complementary one: the input is an \emph{arbitrary} connected graph, which is in general not a Cayley graph, and the contribution is the isometric embedding that places it inside an abelian Cayley host, together with the excursion ratio that quantifies the cost. The two viewpoints meet exactly when $G$\ is already an abelian Cayley graph: there $\varepsilon=1$, the embedding is the identity, and the GE-GFT coincides with the character transform of the intrinsic theories. Our same-filter experiments (Section~\ref{sec:experiments}) confirm this coincidence numerically---on circulant hosts the group-character basis and the Laplacian eigenbasis are statistically indistinguishable, as they must be since a circulant's eigenvectors \emph{are} its group characters. What the embedding adds beyond the intrinsic setting is everything at $\varepsilon<1$: a host distinct from the graph, a complement to be filled (Section~\ref{sec:experiments}), and a transform for graphs that possess no Cayley structure of their own.
		
		\section{Conclusion}
		\label{sec:conclusion}
		We built a harmonic analysis on graphs by transporting signals to a Cayley graph of a finite abelian group into which the graph embeds isometrically. The construction supplies what the spectral GFT cannot: a canonical Fourier basis free of eigenspace ambiguity, a unitary translation obeying an exact group law, a genuine group convolution with the convolution theorem as a theorem, and the classical Plancherel, uncertainty, and sampling identities---exact on structured hosts. A same-filter study showed that the framework's value is this structural exactness, not a denoising-SNR advantage, and revealed the host complement as a usable extension degree of freedom controlled by the excursion ratio. Companion papers develop the embedding construction \cite{FokamP1}, the dimension and order bounds together with the abelian dividend---an exhaustive census showing that compact non-binary hosts are the rule rather than the exception on small graphs \cite{FokamP2}---and the tight wavelet frames \cite{FokamThesis2026} that this Fourier layer supports. Two directions of future work follow directly from Sections~\ref{sec:experiments} and~\ref{sec:discussion}: a case study on a real structured infrastructure network (a sensor grid or a ring-and-spur telecommunication backbone), where the excursion ratio is expected to sit between the torus and the karate-club extremes; and analytic bounds on the repair-round parameter $R$\ for natural graph classes.
		
		
		\section*{Declaration on the Use of Artificial Intelligence}
		In the interest of transparency and research integrity, the authors declare that the artificial-intelligence assistant Claude (Anthropic) was used during the preparation of this manuscript. Its assistance was limited to two roles: (i) support with the implementation, debugging, and reproducibility of the embedding and signal-processing software used to produce the experimental results; and (ii) language and editing support in drafting and refining the manuscript. All mathematical definitions, theorems, proofs, and their verification, together with the conception, scientific direction, and conclusions of this work, are the authors' own. The authors have reviewed the entire manuscript and take full responsibility for its content.
		
		\section*{Acknowledgments}
		The authors thank Solutum Engineering for internet and computing support.
		
		\appendix
		\section{Worked embedding examples}
		\label{app:embed}
		
		We record five embeddings used above, so the construction of Section~\ref{sec:prelim} is concrete and self-contained; the Petersen and Pappus entries make the two most-cited rows of Table~\ref{tab:gallery} reproducible by hand. In each case the host and labeling are verified by exhaustive comparison of all $\binom{n}{2}$\ graph distances against breadth-first search in the host.
		
		\emph{Star $K_{1,3}$\ (binary host).} Each centre--leaf edge is its own class, so the generators are the three leaf labels. The leaves lie pairwise at distance $2$, forcing a sum-free generating set; the smallest binary host is $\Cay(\Z_2^3,\{001,010,100\})$ with the centre at $000$, giving $N=8$\ and $\varepsilon=1/2$\ (Theorem~\ref{thm:star} with $q=3$).
		
		\emph{Petersen graph (composite binary generator).} With the standard numbering (outer $5$-cycle $0\ldots4$, inner pentagram $5\ldots9$, spokes $i\sim i{+}5$), the $\varphi$\ relation partitions the $15$\ edges into exactly five parallel matchings of size $3$: $\{01,38,79\}$, $\{04,27,68\}$, $\{05,23,69\}$, $\{12,49,58\}$, $\{16,34,57\}$. The cycle--class parity matrix has rank $1$\ (its unique relation is $g_1+\cdots+g_5=0$), so $k=5-1=4$: four basis generators plus the composite $(1,1,1,1)$, giving the Clebsch host $\Cay(\Z_2^4,S)$\ with $N=16$, $\varepsilon=0.625$, certified on all $45$\ pairs \cite{FokamP1}.
		
		\emph{Pappus graph (structured partition, $1024\times$\ compaction).} With vertices numbered along the LCF Hamiltonian cycle $[5,7,-7,7,-7,-5]^3$, the $27$\ edges partition into nine pairwise-$\varphi$\ triples (e.g.\ $\{(0,1),(6,7),(12,13)\}$\ and its rotations and chord classes); the $10\times9$\ cycle--class parity matrix has rank $2$, so $k=9-2=7$: seven basis generators and two composites of Hamming weight $5$, host order $2^7=128$\ against the naive $2^{17}$, certified on all $153$\ pairs. Notably, only $2$\ of the $15$\ possible partitions of the Pappus edges into $\varphi$-triples are isometric---the concrete reason the pipeline certifies every candidate rather than trusting $\varphi$-compatibility \cite{FokamP1}.
		
		\emph{Diamond $K_4-e$\ (cyclic factor).} The two degree-2 vertices lie at distance $2$, which no sum-free binary set can realize together with the three unit distances; a $\Z_3$\ factor is forced. The minimal host is the octahedron $\Cay(\Z_2\times\Z_3,\{\pm(0,1),\pm(1,1)\})$, $N=6$, $\varepsilon=2/3$, with vertices labelled $(0,0),(0,1),(0,2),(1,0)$\ (Fig.~\ref{fig:embedding}).
		
		\emph{Path $P_n$\ into a cycle (the large-scale case).} The map $i\mapsto i$\ embeds $P_n$\ isometrically into $C_{2n-2}$: for $0\le i,j\le n-1$\ the cycle distance $\min(|i-j|,\,2n-2-|i-j|)$\ equals $|i-j|$, the path distance, since $|i-j|\le n-1<n-1+1$. Taking the Cartesian product embeds the grid $P_n\times P_n$\ into the torus $C_{2n-2}\times C_{2n-2}$, the construction underlying Section~\ref{sec:largescale}; for $n=128$\ this is the $254\times254$\ torus with $\varepsilon=0.254$.
		

\end{document}